\documentclass[pra,twocolumn,aps,amssymb,amsmath,superscriptaddress,showpacs]{revtex4-1}
\usepackage[latin1]{inputenc}
\usepackage{mathrsfs}
\usepackage{amsmath}
\usepackage{amssymb}
\usepackage{amsthm}
\usepackage{amsfonts}
\usepackage{amstext}
\usepackage{amsopn}
\usepackage{amsxtra}
\usepackage{mathrsfs}
\usepackage{graphicx}

\newtheorem{theorem}{Theorem}
\newtheorem{lemma}{Lemma}

\newcommand{\dps}{\displaystyle}
\newcommand{\ii}{\infty}
\newcommand\R{{\ensuremath {\mathbb R} }}

\newcommand\1{{\ensuremath {\mathds 1} }}

\renewcommand\phi{\varphi}

\newcommand{\cI}{\mathcal{I}}

\newcommand{\cV}{\mathcal{V}}

\newcommand{\eps}{\epsilon}

\newcommand{\bx}{\mathbf{x}}
\newcommand{\bX}{\mathbf{X}}
\newcommand{\by}{\mathbf{y}}
\newcommand{\bz}{\mathbf{z}}
\newcommand{\bu}{\mathbf{u}}
\newcommand{\br}{\mathbf{r}}

\newcommand{\bk}{\mathbf{k}}
\newcommand{\bv}{\mathbf{v}}

\renewcommand{\epsilon}{\varepsilon}

\DeclareMathOperator{\tr}{{\rm Tr}}

\renewcommand{\geq}{\geqslant}
\renewcommand{\leq}{\leqslant}

\renewcommand{\tilde}{\widetilde}

\newcommand{\dx}{{\rm d}\bx}
\newcommand{\dX}{{\rm d}\bX}

\newcommand{\dy}{{\rm d}\by}

\newcommand{\dz}{{\rm d}\bz}

\newcommand{\dk}{{\rm d}\bk}
\newcommand{\dt}{{\rm d}t}

\newtheorem*{fauxthm}{Runge-Gross' uniqueness}

\date{\today}

\begin{document}

\title[Coulomb potentials and Taylor expansions in TDDFT]{Coulomb potentials and Taylor expansions in Time-Dependent Density Functional Theory}

\author{S{\o}ren Fournais}
\affiliation{Department of Mathematics, Aarhus University, Ny Munkegade 118, DK-8000 Aarhus C, Denmark}

\author{Jonas Lampart}
\affiliation{PSL Research University \& CEREMADE (UMR CNRS 7534), University of Paris-Dauphine, Place de Lattre de Tassigny, F-75775 Paris Cedex 16, France}

\author{Mathieu Lewin}
\affiliation{CNRS \& CEREMADE (UMR CNRS 7534), University of Paris-Dauphine, Place de Lattre de Tassigny, F-75775 Paris Cedex 16, France}

\author{Thomas {\O}stergaard S{\o}rensen}
\affiliation{Mathematisches Institut, LMU Munich, Theresienstra{\ss}e 39, D-80333 Munich, Germany}

\begin{abstract}
We investigate when Taylor expansions can be used to prove the Runge-Gross Theorem, which is at the foundation of Time-Dependent Density Functional Theory (TDDFT). We start with a general analysis of the conditions for the Runge-Gross argument, especially the time-differentiability of the density. 
The latter should be questioned in the presence of singular (e.g.~Coulomb) potentials. 
Then, we show that a singular potential in a one-body operator considerably decreases the class of time-dependent external potentials to which the original argument can be applied.
A two-body singularity has an even stronger impact and an external potential is essentially incompatible with it. For the Coulomb interaction and all reasonable initial many-body states, the Taylor expansion only exists to a finite order, except for constant external potentials. Therefore, high-order Taylor expansions are not the right tool to study atoms and molecules in TDDFT.
\end{abstract}

\maketitle


Density Functional Theory (DFT) is one of the pillars of modern quantum chemistry and condensed matter physics. Its time-independent form has been studied in depth and its virtues and limitations are rather well understood~\cite{KohSha-65,Levy-79,Lieb-83b}. On the other hand, Time-Dependent Density Functional Theory (TDDFT) is more recent and certainly less well understood. 

One of the most important steps in the construction of TDDFT is a result by Runge and Gross~\cite{RunGro-84,EngDre-11} which says that the external potential $V(t,\bx)$ in a many-body system is completely determined (up to a constant) by the one-particle density $\rho(t,\bx)$. More precisely, if two potentials $V_1(t,\bx)$ and $V_2(t,\bx)$ give rise to the same one-particle density $\rho(t,\bx)$ for all times $t$, then $V_1(t,\bx)=V_2(t,\bx)+C(t)$. The applicability of this result for Coulomb systems is still under debate and our purpose in this article is to discuss in more detail the possible problems that can arise with Coulombic or more general singular potentials.

The original argument of Runge and Gross relies on the assumption that the external potentials as well as the many-body body wavefunction $\Psi(t)$ are all time-analytic, which means that they can be expanded in a \emph{convergent} Taylor series in powers of $t$. If two potentials $V_1(t,\bx)$ and $V_2(t,\bx)$ give rise to the same $\rho(t,\bx)$, it was argued that each coefficient of $t^k$ in the two Taylor series of $\nabla V_1(t,\bx)$ and $\nabla V_2(t,\bx)$ must be the same. The convergence of these power series to $\nabla V_1(t,\bx)$ and $\nabla V_2(t,\bx)$ then implies that $V_1(t,\bx)=V_2(t,\bx)+C(t)$.

It has recently been noticed~\cite{MaiTodWooBur-10,YanMaiBur-12,YanBur-13} that the time-analyticity of the wavefunction fails in many simple examples, even for time-analytic potentials. This is not surprising, since the \emph{time} regularity of the solution $\Psi(t,\bx)$ to Schrödinger's equation is known to be intimately linked to the \emph{space} regularity of the initial state $\Psi_0(\bx)$, as we will recall. If this initial state $\Psi_0(\bx)$ is not smooth enough with respect to $\bx$, then the resulting solution will not be smooth in $t$. Conversely, even for a very smooth initial state $\Psi_0(\bx)$, non-smooth (e.g. Coulomb) potentials can create singularities that propagate in time and give a Schrödinger solution that is not smooth in $t$. These examples clearly violate the assumptions used in the Runge-Gross approach and it is an open question whether potentials are characterized by the density in those cases. This is of course a fundamental problem for Coulomb potentials as in atoms and 
molecules.

Our purpose in this work is to discuss in detail how this effect arises and what role it plays for Coulomb potentials. 

An important question, rarely discussed in the literature, is to find reasonable assumptions on the potentials and initial state $\Psi_0$ under which the Runge-Gross approach, based on Taylor expansions, is applicable. As explained above, this requires a better understanding of the link between the properties of $\Psi_0(\bx)$ and those of $V(t,\bx)$ that will generate a solution $\Psi(t,\bx)$ that is analytic, or at least smooth, with respect to time. This link is naturally expressed in terms of the spectral theory of the underlying many-body Hamiltonian and it is rather subtle in the presence of singular potentials. We will introduce a precise framework taking care of these questions in Sections~\ref{sec:framework} and~\ref{sec:regularity}.

We then analyze in detail three different situations. In the case of smooth external potentials and interactions, we explain for which $\Psi_0(\bx)$ and $V(t,\bx)$ the original Runge-Gross argument works out and we give a complete, mathematically rigorous, proof in this setting in Section~\ref{sec:RG-smooth}.

We then look at singular potentials. 
In Section~\ref{sec:one-body} we treat two examples of one-particle systems that illustrate how a singularity in the potential considerably reduces the set of allowed time-dependent potentials. 
 A two-body singularity turns out to be even more delicate and an external potential is essentially incompatible with it. Indeed, except for constant external potentials, the many-body wavefunction will in general not be smooth in time for a singular interaction. 

The Coulomb interaction suffers from these difficulties. In Section~\ref{sec:Coulomb} we determine exactly how many time-derivatives make sense in this case, and derive a corresponding finite-order Runge-Gross Theorem. This is not the result that one might have hoped for, but this is certainly the best that can be obtained with an argument based on Taylor expansions.

Another important piece of TDDFT is van Leeuwen's construction~\cite{vanLeeuwen-99} of an external  time-dependent potential $V(t,\bx)$ that produces a given density $\rho(t,\bx)$. His argument is also based on power series expansions in time and thus suffers from the same regularity issues. Several recent works~\cite{RugPenBau-09,RugLee-11,PenRug-11,RugGiePenLee-12,RugPenLee-15} have aimed at justifying the Runge-Gross and van Leeuwen results avoiding the use of Taylor expansions. This is an important and interesting program which has not yet reached a completely satisfactory level of mathematical rigor. We hope that our work will clarify the situation and stimulate further research.

\section{The setting}\label{sec:framework}
In this section, we would like to discuss a general setting, based on physical considerations, for the Runge-Gross Theorem. We consider a system of $N$ particles (fermions or bosons) interacting through a potential $w(\bx-\by)$ and submitted to a fixed external potential $V_0(\bx)$. For atoms and molecules one should consider Coulomb potentials but the situation is kept general for the moment. For simplicity of the exposition, we will discard the spin variable but everything applies \emph{mutatis mutandis} if the particles have an internal degree of freedom. In a time-dependent external potential $V(t,\bx)$, the $N$-body Schrödinger Hamiltonian is
\begin{equation*}
H_{V}=\sum_{j=1}^N-\Delta_{\bx_j}+V_0(\bx_j)+V(t,\bx_j)+\sum_{1\leq j<k\leq N}w(\bx_j-\bx_k)
\end{equation*}
and the corresponding time-dependent Schrödinger evolution equation is
\begin{equation}
\begin{cases}
\dps i\tfrac{\partial}{\partial t}\Psi(t,\bX)=H_V\Psi(t,\bX)\\
\Psi(0,\bX)=\Psi_0(\bX)\,,
\end{cases}
\label{eq:TD-Schrodinger}
\end{equation}
where $\bx_j \in \R^3$ and $\bX=(\bx_1,\dots,\bx_N)$.
We are using units such that $2m=\hbar=1$, and also $4\pi \eps_0=1$. 
For atoms and molecules, $w(\bx-\by)=1/|\bx-\by|$ and the Coulomb potential $-\sum_{m=1}^MZ_m/|\bx-\br_m|$ of the nuclei can be either included in $V_0$, if the nuclei are fixed, or in $V(t,\bx)$ if they move. 
Without the potential $V(t,\bx)$, the Hamiltonian becomes
\begin{equation*}
H_{0}=\sum_{j=1}^N-\Delta_{\bx_j}+V_0(\bx_j)+\sum_{1\leq j<k\leq N}w(\bx_j-\bx_k).
\end{equation*} 
We recall that the density of the $N$-particle solution is defined by 
$$\rho(t,\bx)=N\int_{\R^{3(N-1)}}|\Psi(t,\bx,\bx_2,...,\bx_N)|^2\,\dx_2\cdots \,\dx_N.$$

For the Runge-Gross Theorem, we have to specify a class of initial conditions $\cI=\{\Psi_0\text{'s}\}$ as well as a class of considered external potentials $\cV=\{V(t,\bx)\text{'s}\}$ (defined on a given time interval $[0,t_{\rm max})$). The goal is to find the sets $\cI$ and $\cV$ for which the following theorem is valid.

\begin{fauxthm}
Let $V_j(t,\bx)$, $j=1,2$, be two potentials in the class $\cV$ and let $\Psi_j(t,\bx_1,...,\bx_N)$ be the corresponding solutions to~\eqref{eq:TD-Schrodinger}, with the same initial state $\Psi_0\in\cI$. If the associated densities satisfy $\rho_1(t,\bx)=\rho_2(t,\bx)$ for all $t\in[0,t_{\rm max})$ and all $\bx\in\R^3$, then $V_1(t,\bx)=V_2(t,\bx)+C(t)$.
\end{fauxthm}

Depending on the physical context, the set $\cV$ could be chosen very small, for instance it could consist of the Coulomb potential generated by one moving nucleus, with the position $\br(t)$ of this nucleus being the only parameter:
$$\cV=\left\{-\frac{Z}{|\bx-\br(t)|},\ \br(t)\in\R^3\right\}.$$
It could also be very large and contain a whole class of smooth functions of the time and space variables. In general it is desirable to have both sets $\cI$ and $\cV$ as large as possible. However they cannot be chosen independently. Uniqueness might hold for a very large set $\cV$ provided that $\cI$ is very small, and conversely. Furthermore, the choice of $\cI$ and $\cV$ could (and will) highly depend on the properties of the potentials $w$ and $V_0$. Since in practice the exact wavefunction $\Psi_0$ is unknown, we believe that it is appropriate to choose $\cI$ as large as possible. 

We think that the following reasonable conditions should be imposed:
\begin{enumerate}
\item[(H1)] If the initial datum $\Psi_0(\bX)$ is in $\cI$ and $V(t,\bx)$ is a potential in $\cV$, then the wavefunction $\Psi(t,\bX)$, solution to the Schrödinger equation~\eqref{eq:TD-Schrodinger}, must be in $\cI$ for all later times $t$.
\item[(H2)] The constant potentials $V(t,\bx)=C(t)$ all belong to $\cV$.
\item[(H3)] If $V(t,\bx)$ is in $\cV$, then the time-independent potential $V(t_0,\bx)$ is in $\cV$ as well, for every $t_0\in[0,t_{\rm max})$.
\item[(H4)] If $V(\bx)\in\cV$ is time-independent, then all the eigenfunctions of $H_V$ are in $\cI$.
\end{enumerate}
Condition (H1) is here to be able to apply the Runge-Gross argument at later times $t_0>0$, with new initial state $\Psi(t_0,\bX)$. There is no reason to give $t=0$ a special role. On the other hand, (H2) covers the case when no external field is applied to the system ($V=0$). By a gauge transformation, all the constants $C(t)$ must then be allowed in $\cV$. It is often assumed that the evolution before the considered time $t=0$ was governed by a fixed time-independent potential, which is only changed at positive times. It is then natural to assume, as in (H3), that time-independent potentials are in $\cV$ as well. Finally, the assumption (H4) is here to make sure that the usual time-independent DFT is covered by the theory. Indeed, in applications $\Psi_0$ will often be the ground state of the potential $V(0,\bx)$, and $V(t,\bx)$ would then be chosen in order to bring the system to another interesting state.

We remark that $\cI$ depends on the number of particles $N$, whereas we could also demand that $\cV$ is independent of $N$, if needed. This may of course result in a smaller set, depending on the situation.

To our knowledge the choice of the sets $\cV$ and $\cI$ has never been discussed in the literature. This is not a question of purely mathematical interest since the applicability of the Runge-Gross Theorem to physical systems relies on these sets not being too small. In the following we will give several examples that will clarify the respective roles of $\cI$ and $\cV$. In the next section we start by discussing for which choices of $\cI$ the wavefunction is smooth in time, for a time-independent potential $V(\bx)$, before addressing the more general case of time-dependent potentials $V(t,\bx)$.

\section{Time-regularity of solutions to the Schrödinger equation}\label{sec:regularity}

In this section we would like to recall when a solution $\Psi(t)$ of Schrödinger's equation depends smoothly on the time $t$. Based on these results, we propose an abstract definition of the sets $\cI$ and $\cV$, which provides a smooth-in-time solution $\Psi$, and for which the original Runge-Gross argument is thus applicable. These sets will be made more explicit in particular examples in the following sections. 

We start by discussing the case of time-independent potentials that must all belong to $\cV$ by Condition (H2). We then turn to the general time-dependent case.

\subsection{The time-independent case}
Let $V(\bx)$ be a time-independent potential that makes $H_V$ a self-adjoint operator (more precisely, we assume that the potential energy is infinitesimally $H_0$-bounded, see the discussion below). Then the solution to Schrödinger's equation~\eqref{eq:TD-Schrodinger} exists for any state $\Psi_0$, $\int_{\R^{3N}}|\Psi_0|^2=1$, and can be expressed in the abstract form
$$\Psi(t,\bX)=\big(e^{-itH_V}\Psi_0\big)(\bX)\,.$$
The operator $e^{-itH_V}$ is unitary and it is defined using the so-called functional calculus for self-adjoint operators~\cite{Davies}. 
In general $e^{-itH_V}\Psi_0$ is \emph{not} given by the power series $\sum_{k\geq 0}(-itH_V)^k\Psi_0/k!$. The reason is that, for most $\Psi_0$'s, $(H_V)^k\Psi_0$ does not make sense. 
Indeed, if the series converges, then $e^{t H_V}\Psi_0=\sum_{k\geq 0}(tH_V)^k\Psi_0/k!$ is also convergent and gives a solution to the ``backward heat equation''
$$\frac{\partial}{\partial t}\Psi-H_V\Psi=0$$
which is well-known to be ill-posed for most initial states $\Psi_0$.

The set of functions $\Psi$ for which $H_V\Psi$ is square-integrable is called the \emph{domain of $H_V$} and is often denoted by $D(H_V)$. For instance, if $T=-\Delta$ is the kinetic energy operator, then $\psi\in D(T)$ if and only if $\|T\psi\|^2=\int|\Delta\psi(\bx)|^2\,\dx=\int |\bk|^4|\widehat{\psi}(\bk)|^2\,\dk$ is finite, where $\widehat{\psi}(\bk)$ is the Fourier transform of $\psi(\bx)$. 

When looking at the power series, it is natural to think that $\Psi(t)=\exp(-itH_V)\Psi_0$ will be $k$ times differentiable in $t$, if and only if $(H_V)^k\Psi_0$ is square integrable, and that then
$$\frac{\partial^k}{\partial t^k}\Psi(t,\bX)=(-i)^k\left(e^{-itH_V}(H_V)^k\Psi_0\right)(\bX).$$
This intuitive result is true~\cite[Theorem VIII.7]{ReeSim1}, but one has to be very careful with how $(H_V)^k$ is defined.

Before discussing the definition of $(H_V)^k$, we remark that $t\mapsto \Psi(t,\bX)$ is time-analytic if and only if
$\sum_{k\geq0}\|(H_V)^k\Psi_0\|R^k/k!$ is finite for some $R>0$, and this is the only case for which $\Psi(t,\bX)$ can be reconstructed from its Taylor series 
\begin{equation*}
\Psi(t,\bX)=\sum_{k\geq 0}\frac{(-iH_V)^k\Psi_0(\bX)}{k!}\,t^k,
\end{equation*}
for $|t|\leq R$. Such states $\Psi_0$ are called \emph{analytic vectors of $H_V$}, a concept that was introduced by E.~Nelson in~\cite{Nelson-59} and played an important role in quantum field theory.

The meaning of the operator $(H_V)^k$, from now on denoted $H_V^k$, is again determined by the functional calculus. In more practical terms, $H_V^k\Psi_0$ is defined by \emph{recursively} calculating $H_V \Psi_0$, $H_V(H_V\Psi_0)$, etc., which must all be square-integrable functions. The domain of $H_V^k$ is, therefore,
\begin{multline}
D(H_V^k)=\bigg\{\Psi\in D(H_V)\ :\ H_V\Psi\in D(H_V),\\ H_V(H_V\Psi)\in D(H_V),\ \cdots,\ H_V^{k-1}\Psi\in D(H_V)\bigg\}.
\label{eq:domain_H_k}
\end{multline}
The space $D(H_V^k)$ can be shown to be invariant under the flow $\exp(-itH_V)$. It contains all the eigenfunctions of $H_V$, since $H_V\Psi=\lambda\Psi\in D(H_V)$.

The operator $H_V^k$ cannot be fully understood without identifying precisely its domain $D(H_V^k)$. But computing the domain $D(H_V^k)$ can sometimes be a rather difficult task. For instance, for the hydrogen atom which will be considered in more detail in Section~\ref{sec:one-body}, the domain $D(h)$ of $h=-\Delta-1/|\bx|$ is the same as $D(-\Delta)$. However, for sufficiently large $k$ the domain $D(h^k)$ becomes different from $D(-\Delta^k)$. It contains additional boundary conditions on the derivatives of $\psi$ at $\bx=0$. After applying $h$ too many times, the resulting wavefunction $h^{k-1}\psi$ will be singular at $\bx=0$ for a smooth $\psi$.

The operator $H_V^k$ can therefore not easily be understood by looking at the formula for $H_V$. For singular potentials there are usually many consistent choices of boundary conditions for the expression obtained by calculating the $k$th power of $H_V$. Only one of these characterizes the domain of $H_V^k$, that is, the correct boundary conditions arising from the constraints that $H_V^j\Psi_0\in D(H_V)$ for $j=1,...,k-1$.

Let us clarify this by an example. In~\cite{MaiTodWooBur-10} the authors considered the one-dimensional state $\psi_0(x)=\kappa e^{-|x|}$. This state is not in the domain $D(T)$ of the free kinetic energy operator $T=-{\rm d}^2/{\rm d}x^2$. Indeed, functions in $D(T)$ can be shown to be differentiable with a continuous derivative and $\psi_0$ does not have a continuous derivative. In particular, $\psi(t)=\exp(-itT)\psi_0$ is not differentiable in time. On the other hand, $({\rm d}^2/{\rm d}x^2)^k\psi_0=\psi_0^{(2k)}=\psi_0$ has a meaning outside of zero for all $k\geq1$, but the latter expression cannot be the formula of $T^k\psi_0$ since $\psi_0\notin D(T^k)$. This is why $\psi(t)=\exp(-itT)\psi_0$ does not coincide with the series $\sum_{k\geq0}\psi_0^{(2k)}(-it)^k/k!=\kappa e^{-it-|x|}$, as was observed by the authors of~\cite{MaiTodWooBur-10}. The deeper reason for this discrepancy is that $D(T)$ contains the boundary condition $\psi'(0^-)=\psi'(0^+)$ whereas the equation $-\psi_0''=\psi_0$ holds with the 
different boundary condition $\psi'(0^+)-\psi'(0^-)=-2\psi(0)$.

Let us now return to the question of choosing the sets $\cI$ and $\cV$. In order to be able to apply the original Runge-Gross argument which relies on the time-differentiability of $\Psi(t)$, the previous discussion tells us that we should at least choose a set $\cI$ that satisfies
$$\cI\subset D(H_V^k)$$
for all time-independent potentials $V(\bx)\in\cV$ and all $k\geq1$. 
Doing so will imply that $\Psi(t,\bX)$ is infinitely differentiable in $t$, for any $\Psi_0\in\cI$. If analyticity is to be required, then $\cI$ must be included in the set of analytic vectors of $H_V$ for all $V(\bx)$. Now, the simplest choice is to take $\cI$ as the intersection of all these spaces but $\cI$ could then be extremely small. More severely, this set would not be invariant under the flows corresponding to all the possible $H_V$'s, violating our condition (H1), if the domains $D(H_V^k)$ are truly different when $V(\bx)$ is varied. For this reason, it is natural to restrict ourselves to the potentials $V(\bx)$ that preserve the domain of $H_0^k$ for all $k$:
\begin{equation}
D(H_V^k)=D(H_0^k),\qquad \forall k\geq1.
\label{eq:constraint}
\end{equation}
For such a potential, and any $\Psi_0$ in
\begin{equation}
\boxed{\cI=\bigcap_{k\geq1} D(H_0^k),}
 \label{eq:defI}
\end{equation}
the solution $\Psi(t)$ is differentiable infinitely often in time, and $\Psi(t)$ stays in $\cI$ for all later times $t>0$.

In this way, the choice of the $\Psi_0$'s which give a time-differentiable solution $\Psi(t)$ has led us to the constraint~\eqref{eq:constraint} on the set of potentials $\cV$. The property~\eqref{eq:constraint} is however not so easy to check for a given potential $V(\bx)$. Fortunately, there is an equivalent formulation which we are going to discuss now.
 
Let us recall that $\sum_{j=1}^NV(\bx_j)$ is bounded relative to $H_0$, with infinitesimal bound, if 
\begin{multline}
\int_{\R^{3N}}\bigg|\sum_{j=1}^N V(\bx_j)\Psi(\bX)\bigg|^2\dX\\
\leq  \epsilon^2\int_{\R^{3N}}\left|(H_0+C_\epsilon)\Psi(\bX)\right|^2\dX
\label{eq:Kato}
\end{multline}
for every $\epsilon>0$ and $\Psi\in D(H_0)$ (see~\cite[Sec.~1.4]{Davies}). This is a famous condition, satisfied for Coulomb potentials as was shown by Kato in the 50s~\cite{Kato-51}.

\begin{lemma}[Condition on $V(\bx)$]\label{lem:condition_V_preserve}
Assume that $H_0$ is bounded below and that the time-independent potential $V(\bx)$ satisfies~\eqref{eq:Kato}. Then $V(\bx)$ satisfies the constraint~\eqref{eq:constraint} if and only if for every $\Psi\in D(H_0^k)$, the function $\sum_{j=1}^NV(\bx_j)\Psi(\bX)$ is in $D(H_0^{k-1})$ for every $k\geq1$. In other words, the multiplication operator by $\sum_{j=1}^NV(\bx_j)$ maps $D(H_0^k)$ to $D(H_0^{k-1})$.
\end{lemma}
\begin{proof}
Assume first that~\eqref{eq:constraint} is satisfied. Then $H_V=H_0 + \sum_{j=1}^NV(\bx_j)$ maps, by definition of these domains, $D(H_V^k)=D(H_0^k)$ to $D(H_V^{k-1})=D(H_0^{k-1})$. Obviously, $H_0$ also has this property, so $H_V-H_0=\sum_{j=1}^NV(\bx_j)$ maps $D(H_0^k)$ to $D(H_0^{k-1})$, for every $k\geq 1$.

For the converse implication, we recall that the infinitesimal relative bound~\eqref{eq:Kato} ensures that $H_V$ is self-adjoint on the domain $D(H_V)=D(H_0)$ (see~\cite[Theorem 1.4.2]{Davies}). 
Now assume that $\sum_{j=1}^NV(\bx_j)$ maps $D(H_0^{2k+1})$ to $D(H_0^{2k})$. This map is then automatically continuous by the closed graph theorem~\cite[Theorem III.12]{ReeSim1}, because continuity from $D(H_0)$ to the space $L^2(\R^{3N})$ of square-integrable functions (assumption~\eqref{eq:Kato}) implies that the graph of its restriction is closed in $D(H_0^{2k+1})\times D(H_0^{2k})$.
Consequently, $\sum_{j=1}^NV(\bx_j)(H_0 + C_\eps)^{-1}$ maps $D(H_0^{2k})$ continuously to itself.
On the other hand,~\eqref{eq:Kato} means that $\sum_{j=1}^NV(\bx_j)(H_0 + C_\eps)^{-1}$ is norm-bounded by $\epsilon$ on $L^2(\R^{3N})$. By the theory of interpolation~\cite[Prop.~9 Chap.~IX.4]{ReeSim2} it is therefore bounded by a quantity proportional to $\sqrt{\eps}$ on $D(H_0^k)$. After choosing $\eps$ sufficiently small, the same argument as for the Kato-Rellich theorem (that we used for $k=0$) applies~\cite[Sec.~1.4]{Davies}, and shows that $D(H_V^k)=D(H_0^k)$.
\end{proof}

If the domain $D(H_0^k)$ contains specific boundary conditions at the singularities of the potentials $w$ and $V_0$, then the lemma says that multiplying by the function $\sum_{j=1}^NV(\bx_j)$ must preserve these conditions. As we will show in Section~\ref{sec:singular}, this is a very restrictive condition.

\subsection{The time-dependent case}

For a time-dependent potential $V(t,\bx)$, the time derivatives of $\Psi(t)$ will clearly involve time-derivatives of $V$. Therefore, in view of the discussion of the previous section, a natural condition on $V(t,\bx)$ is to assume that $\sum_{j=1}^N\partial_t^\ell V(t,\bx_j)$ exists and satisfies the condition of Lemma~\ref{lem:condition_V_preserve}, that is,
\begin{multline}
\int_{\R^{3N}}\bigg|H_0^{k-1} \sum_{j=1}^N\partial_t^\ell V(t,\bx_j)\Psi(\bX)\bigg|^2\dX\\
\leq C_{k,\ell,t_{\mathrm{max}}} \int_{\R^{3N}}\left|\left(H_0^k\Psi\right)(\bX)\right|^2 + |\Psi(\bX)|^2\dX
\label{eq:constraint_TD}
\end{multline}
for every $k\geq1$, every $\ell\geq0$, every $t\in[0,t_{\rm max})$ and every $\Psi\in D(H_0^k)$. 
Additionally, we have the condition~\eqref{eq:Kato} that $\sum_{j=1}^NV(t, \bx_j)$ should be infinitesimally bounded with respect to $H_0$. This now defines us a set 
\begin{equation}
\boxed{\cV=\Big\{V(t,\bx)\text{ satisfying~\eqref{eq:Kato} and~\eqref{eq:constraint_TD}} \Big\}.}
\label{eq:defV}
\end{equation}

For these choices of $\cI$ in~\eqref{eq:defI} and $\cV$ in~\eqref{eq:defV}, it remains to check if the conditions (H1)--(H4) are satisfied. (H2) and (H3) follow directly from the definition of $\cV$. Using that eigenfunctions belong to $D(H_V^k)=D(H_0^k)$ for all $k\geq1$, (H4) follows immediately. Only (H1) needs a more careful treatment. The following result shows that it indeed holds.

\begin{theorem}[Regularity of solutions to the time-dependent Schrödinger equation]\label{thm:regularity}
Let $\Psi_0\in \cI$ and $V(t,\bx)\in\cV$. Then the solution $\Psi(t,\bX)$ to Schrödinger's equation~\eqref{eq:TD-Schrodinger} belongs to $\cI$ for every $t\in[0,t_{\rm max})$. It is differentiable in $t$ infinitely many times, and all its $t$-derivatives also belong to $\cI$.
\end{theorem}
\begin{proof}
It was proved by Kato~\cite[Theorem 4]{Kato-53} that if the domain $D(H_{V(t_0)})=D(H_0)$ is independent of $t_0\in [0,t_{\rm max})$ and $\partial_t H_{V}= \sum_{j=1}^N \partial_t V(t,\bx_j)$ is a continuous map from $D(H_0)$ to the Hilbert space of square-integrable functions, then the solution $\Psi(t,\bX)$ to Schrödinger's equation~\eqref{eq:TD-Schrodinger} exists, and it belongs to $D(H_0)$ if the initial condition $\Psi_0(\bX)$ is an element of $D(H_0)$.

In our case, the conditions on the domain and the initial condition are implied by the choices of $\cV$ and $\cI$. Continuity of $\partial_t H_{V}$ is exactly the statement of condition~\eqref{eq:constraint_TD} with $\ell=k=1$. We thus have $\Psi(t,\bX)\in D(H_0)$ for $t\in [0,t_{\rm max})$.

If we consider $D(H_0)$ as a Hilbert space, the domain of $H_{V(t_0)}$ on this space is exactly given by $D(H_0^2)$ (see Equation~\eqref{eq:domain_H_k}). 
Using that $\Psi_0(\bX)\in D(H_0^2)$ and condition~\eqref{eq:constraint_TD} (with $\ell=1$, $k=2$) we can apply Kato's result in this setting, and find that $\Psi(t,\bX)$ is an element of $D(H_0^2)$ for all $t\in [0,t_{\rm max})$. 
Hence $H_V \Psi(t,\bX)$ is an element of $D(H_0)$. To prove that $H_V \Psi(t,X)$ is differentiable in time one uses~\eqref{eq:constraint_TD} again. This shows that $\Psi(t,\bX)$ has two time-derivatives, and repeating this argument shows that it has infinitely many.
\end{proof}

In this section we have defined the sets $\cI$ and $\cV$ such as to be able to take as many time derivatives as we like, and the original proof of Runge and Gross will apply. Of course, the main question is to identify more precisely the sets $\cI$ and $\cV$ in order to understand which physical systems are covered. It is the potentials $w$ and $V_0$ appearing in the definition of $H_0$ that will determine the properties of these sets.

We remark that the exact same construction can be applied to obtain a $\Psi(t)$ which has only $m$ time-derivatives. In this case, $\cI=D(H_0^m)$ and the conditions~\eqref{eq:constraint_TD} only have to be verified for $k+\ell\leq m+1$. We will come back to this generalization later in Section~\ref{sec:Coulomb}, when treating Coulomb interactions.

\section{The smooth case}\label{sec:RG-smooth}

In this section we study the case when $w$ and $V_0$ are smooth, and identify explicitly the two sets $\cI$ and $\cV$. We will see that the corresponding wavefunction $\Psi(t)$ is also smooth, and that the original Runge-Gross method can be followed without any danger. 

We first consider the kinetic energy operator $T=\sum_{j=1}^N-\Delta_{\bx_j}$. Using the Fourier transform, one sees that a square-integrable function $\Psi$ belongs to $D(T^k)$ if and only if $\int_{\R^{3N}}|{\bf K}|^{4k}|\widehat{\Psi}({\bf K})|^2\, {\rm d}{\bf K}$ is finite. Said differently, a function will be in all the domains $D(T^k)$ with $k\geq1$, if its Fourier transform decays faster than any polynomial. By the theory of Sobolev~\cite{LieLos-01, Brezis-10}, this turns out to imply that $\Psi$ has infinitely many space derivatives. On the other hand, $\Psi$ will be an analytic vector of $T$ if its Fourier transform decays exponentially and this implies that it is a real-analytic function of $\bX$. 

We now assume that $V_0$ and $w$ are smooth functions, that is, they are differentiable infinitely many times and their derivatives are all uniformly bounded on $\R^3$. Under this condition, one can show that the domain of $H_0^k$ is equal to that of $T^k$. The argument is exactly the same as in the proof of Lemma~\ref{lem:condition_V_preserve}. Therefore, the set $\cI$ defined in~\eqref{eq:defI} is just the space of square-integrable, fermionic or bosonic, functions whose Fourier transform decays faster than any polynomial.

The set $\cV$ defined in~\eqref{eq:defV} consists of the potentials $V(t,\bx)$ such that $\sum_{j=1}^N\partial_t^\ell V(t,\bx_j)$ maps $\cI$ into itself, for every $\ell$. This shows that $\partial_t^\ell V$ must be smooth. Using the translation-invariance of $T$, $V(t,\bx)$ and its derivatives can also be shown to be uniformly bounded on $\R^3$. Therefore, Theorem~\ref{thm:regularity} can be rephrased in the more explicit form:

\begin{theorem}[Regularity for smooth potentials]\label{thm:regularity_smooth}
Assume that $V_0$ and $w$ are differentiable infinitely often in space and that their derivatives are all uniformly bounded on $\R^3$. Assume also that $V(t,\bx)$ is differentiable infinitely often in space-time and its derivatives are bounded uniformly with respect to $\bx$ for every $t$. Then, for every $\Psi_0\in\cI$, the solution $\Psi(t,\bX)$ to~\eqref{eq:TD-Schrodinger} is smooth in space-time and $\partial_t^\ell\Psi(t,\bX)$ belongs to $\cI$ for all $\ell\geq0$ and all $t\in [0,t_{\rm max})$.
\end{theorem}

Based on this result, we can now give a complete proof of the Runge-Gross Theorem in the smooth case.

\begin{theorem}[Runge-Gross for smooth potentials]\label{thm:RG_smooth}
Let $\Psi_0$, $V_0$, $w$ and $V=V_m$, $m=1,2$ be as in Theorem~\ref{thm:regularity_smooth}. 
Let $\Psi_m(t,\bX)$ be the solution to~\eqref{eq:TD-Schrodinger} with $V=V_m$ and let $\rho_m(\bx)$ be its density.
If $\rho_1(t,\bx)=\rho_2(t,\bx)=:\rho(t,\bx)$ for all $t\in[0,t_{\rm max})$ and all $\bx\in\R^3$, then 
\begin{equation}
\rho(0,\bx)\nabla\frac{\partial^\ell}{\partial t^\ell}(V_1-V_2)(0,\bx)=0 
 \label{eq:identification_derivatives}
\end{equation}
for all $\ell\geq0$ and all $\bx\in\R^3$. If the nodal set $\{\bx\in\R^3\ :\ \rho(0,\bx)=0\}$ has zero volume, then 
$$\frac{\partial^\ell}{\partial t^\ell}(V_1-V_2)(0,\bx)=c_\ell$$
for all $\ell\geq0$ and all $\bx\in\R^3$.

If in addition $V_1(t,\bx)-V_2(t,\bx)$ is analytic in $t$ for every fixed $\bx\in\R^3$, then 
$$V_1(t,\bx)-V_2(t,\bx)=C(t)=\sum_{\ell\geq0}\frac{c_\ell}{\ell!} t^\ell.$$ 
\end{theorem}


We emphasize that only the time-analyticity of $V_1-V_2$ is required in order to get the Runge-Gross result. It is \emph{not} needed that $\Psi(t)$ is itself time-analytic, as is sometimes stated in the literature. 
We now sketch a proof of the theorem. The reasoning is essentially the same as in~\cite{RunGro-84}, although we do not use the current density $\mathbf{j}(t,\bx)$.

\begin{proof}[Proof of Theorem~\ref{thm:RG_smooth}]
Let $\phi(\bx)$ be an arbitrary (differentiable and bounded) function. We clearly have
\begin{equation*}
 \int_{\R^3} \rho_m(t,\bx)\phi(\bx) \dx= N\langle \Psi_m(t)\vert \phi(\bx_1)\Psi_m(t)\rangle\,.
\end{equation*}
The second time-derivative of this equation is given by
\begin{align*}
 &\frac{{\rm d}^2}{\dt^2}\int_{\R^3} \rho_m(t,\bx)\phi(\bx) \dx\\
&\qquad=-N\Big\langle \Psi_m(t)\Big\vert [H_0, [H_0, \phi(\bx_1)]]\Psi_m(t)\Big\rangle\\ 
&\qquad\qquad+N\Big\langle\Psi_m(t)\Big\vert[V_m(t,\bx_1),[\Delta_{\bx_1}, \phi(\bx_1)]]\Psi_m(t)\Big\rangle\,.
\end{align*}
Since $\Psi_m(t=0,\bX)=\Psi_0(\bX)$, the first term does not depend on $m$ at time $t=0$.
A simple computation gives
$$\big[V_m(t,\bx_1),[\Delta_{\bx_1}, \phi(\bx_1)]\big]=-2\nabla V_m(t,\bx_1)\cdot\nabla\phi(\bx_1).$$
 Assuming that $\rho_1(t,\bx)=\rho_2(t,\bx)$ then yields
\begin{align*}
0&=\frac{{\rm d}^2}{\dt^2}\int_{\R^3} (\rho_1-\rho_2)(t,\bx)\,\phi(\bx)\,\dx\bigg\vert_{t=0}\\
%
&=-2N  \int_{\R^3} \rho(0,\bx) \nabla\phi(\bx) \cdot \nabla\left(V_1(0,\bx)-V_2(0,\bx)\right) \,\dx\,.
\end{align*}
Choosing now $\phi(\bx)=V_1(0,\bx)-V_2(0,\bx)$, we obtain
\begin{equation}\label{eq:deriv_V^2}
 \int_{\R^3} \rho(0,\bx) |\nabla (V_1-V_2)(0,\bx)|^2 \,\dx=0\,.
\end{equation}
Because $\rho(0,\bx)\geq 0$ we then find that $\nabla (V_1-V_2)(0,\bx)$ must be zero at every point where $\rho(0,\bx)$ does not vanish, which implies~\eqref{eq:identification_derivatives} for $\ell=0$. This proves that $V_1(0,\bx)-V_2(0,\bx)$ is constant on each connected component of the set $\{\rho(0,\bx)>0\}$. If $\rho(0,\bx)$ never vanishes, then clearly $V_1(0,\bx)-V_2(0,\bx)=c_0$. In fact, this also holds if the nodal set $\{\rho(0,\bx)=0\}$ has zero volume, since $V_1(0,\bx)-V_2(0,\bx)$ is continuous by assumption. 

To obtain~\eqref{eq:identification_derivatives} for $\ell=1$ we take the third time-derivative of $\langle \Psi_m(t)\vert \phi(\bx_1)\Psi_m(t)\rangle$ at $t=0$, which yields
\begin{align}\label{eq:3rd derivative}
 &\frac{{\rm d}^3}{\dt^3}\int_{\R^3} \rho_m(t,\bx)\phi(\bx) \dx\bigg\vert_{t=0}\\
&\qquad=-Ni\Big\langle \Psi_0\Big\vert [H_0,[H_0, [H_0, \phi(\bx_1)]]]\Psi_0\Big\rangle\notag\\ 
&\qquad\qquad-Ni\sum_{j=1}^N\Big\langle \Psi_0\Big\vert[V_m(0,\bx_j),[H_0, [H_0, \phi(\bx_1)]]]\Psi_0\Big\rangle\notag\\ 
&\qquad\qquad+Ni\Big\langle\Psi_0\Big\vert[H_0,[V_m(0,\bx_1),[\Delta_{\bx_1}, \phi(\bx_1)]]]\Psi_0\Big\rangle\notag\\
&\qquad\qquad-2N  \int_{\R^3} \rho(0,\bx) \nabla\phi(\bx) \cdot \nabla \frac{\partial}{\partial t}V_m(0,\bx)\,\dx\,.\notag
\end{align}

The operator $H_0$ and its (iterated) commutators with $\phi(\bx_1)$ and $V_m(0,\bx_j)$ are differential operators, that is, they can be written in the form
\begin{equation}
D=\sum_{0\leq |j_1|,...,|j_N|\leq M}f_{j_1,...,j_N}(\bx_1,...,\bx_N)\,\partial_{\bx_1}^{j_1}\cdots \partial_{\bx_N}^{j_N} 
 \label{eq:diff_operator}\,,
\end{equation}
with smooth coefficient functions $f_{j_1,...,j_N}(\bx_1,...,\bx_N)$.

We now make the following observation: If $D_1, D_2, D_3$ are such differential operators, then
\begin{multline}\label{eq:commutatorsV1=V2}
 \langle \Psi_0\vert D_1[V_1(0,\bx_j),D_2]D_3\Psi_0\rangle\\= \langle \Psi_0\vert D_1[V_2(0,\bx_j),D_2]D_3\Psi_0\rangle\,,
 \end{multline}
 for any $j=1,\dots, N$.
As commuting $V_1(\bx_j)$ with a differential operator yields a derivative of $V_1(\bx_j)$ multiplied (from the left and right) by differential operators, it suffices to show that
\begin{equation}\label{eq:derivV1=V2}
 \langle \widetilde{D}_1\Psi_0\vert \nabla (V_1-V_2)(0,\bx_j)\widetilde{D}_2\Psi_0\rangle=0
\end{equation}
for any such differential operators $\widetilde{D}_1,\widetilde{D}_2$. 
By~\eqref{eq:deriv_V^2} and continuity of $V(0,\bx)$, the open set of points where $\nabla (V_1-V_2)(0,\bx)$ is non-zero is contained in the nodal set of $\rho(0,\bx)$.
Because we are considering either fermions or bosons, $\rho(0,\bx)=0$ implies that $\Psi_0(\bx_1,  \dots, \bx_N)=0$ if $\bx_j=\bx$ for some $j=1,..,N$. Therefore, $\widetilde{D}_1\Psi_0(\bX)=\widetilde{D}_2\Psi_0(\bX)=0$ on the open set where $\nabla (V_1-V_2)(0,\bx_j)$ does not vanish, so~\eqref{eq:derivV1=V2} holds.
%

Now, using the observation~\eqref{eq:commutatorsV1=V2} we may replace $V_1(0,\bx)$ by $V_2(0,\bx)$ in all the commutator expressions arising in the third time-derivative~\eqref{eq:3rd derivative}. When taking the difference, these terms then cancel and only the last term remains:
\begin{align*}
0&= \frac{{\rm d}^3}{\dt^3}\int_{\R^3} (\rho_1-\rho_2)\phi(t,\bx)\,\dx\bigg\vert_{t=0}\\
&=-2N  \int_{\R^3} \rho(0,\bx) \nabla\phi(\bx) \cdot \nabla\frac{\partial}{\partial t} \left(V_1-V_2\right)(0,\bx)\,\dx\,.
\end{align*}
This gives~\eqref{eq:identification_derivatives} with $\ell=1$ by choosing $\varphi(\bx)=\frac{\partial}{\partial t} \left(V_1-V_2\right)(0,\bx)$. From this we deduce that an equation like~\eqref{eq:commutatorsV1=V2} also holds for the time-derivatives $\frac{\partial}{\partial t}V_1(0,\bx)$, $\frac{\partial}{\partial t}V_2(0,\bx)$, and this allows us to obtain~\eqref{eq:identification_derivatives} for $\ell=2$, and recursively for all $\ell$.

If the nodal set of $\rho(0,\bx)$ has zero volume, then $\frac{\partial^\ell}{\partial t^\ell}(V_1-V_2)(0,\bx)$ equals some constant $c_{\ell}$. Finally, if $V_1(t,\bx)-V_2(t,\bx)$ is analytic in $t$, then the Taylor series
\begin{equation*}
\sum_{\ell\geq0}\frac{c_\ell}{\ell!} t^\ell = V_1(t,\bx)-V_2(t,\bx)
 \end{equation*}
converges, and is clearly independent of $\bx$. 
\end{proof}

\section{Singular potentials}\label{sec:singular}

\subsection{Singular potentials in one-body operators}\label{sec:one-body}

In this section, we consider the case $N=1$ with a singular potential $V_0$. By looking at two examples (first $V_0(x)=\lambda\delta(x)$ in 1D, and second the hydrogen atom), we find that all the derivatives of the potentials $V(t,\bx)$ in $\cV$ must vanish at the singularities of $V_0$. In other words, these potentials are so flat that they hardly influence the dynamics close to the singularities. 
We see that the presence of a singular potential $V_0$ has considerably reduced the set $\cV$.

\subsection*{A delta potential in 1D}
We now study the case of
\begin{equation}
H_0=-\frac{\rm d^2}{{\rm d}x^2}+\lambda \delta(x)
\end{equation}
with $x\in\R$. The domain of $H_0$ is the set of square-integrable functions $\psi(x)$ such that (cf.~\footnote{\label{note:weak_deriv} Here, $\psi''$ denotes the weak derivative (cf.~\cite{Brezis-10}) of $\psi$ on the respective interval. Evaluation of a function at $0^\pm$ stands for its limit at zero from the right ($0^+$) respectively left ($0^-$). These limits exist for $\psi$ and $\psi'$, given that $\psi''$ is square-integrable~\cite[Theorem~8.2]{Brezis-10}.})
$$\begin{cases}
\dps  \int_0^\ii |\psi''(x)|^2\,dx\text{ and }\int_{-\ii}^0 |\psi''(x)|^2\,dx\  \text{are finite}
\\
\psi(0^-)=\psi(0^+),\\
\psi'(0^+)-\psi'(0^-)=\lambda \psi(0). 
\end{cases}
$$
\begin{theorem}[Potentials for a delta interaction]\label{thm:delta_1D}
Assume that $\lambda\neq0$ and let $V(x)$ a be real-valued, measurable function on $\R$. 
Then, $V$ satisfies condition~\eqref{eq:constraint} if and only if $V$ has infinitely many continuous and bounded derivatives on $\R$ and at the origin satisfies 
$$\frac{{\rm d}^k}{{\rm d}x^k}V(0)=0,\quad\text{for all } k\geq1.$$
\end{theorem}

With our definition~\eqref{eq:defV}, the set $\cV$ of allowed potentials then only contains functions $V(t,x)$ that are smooth in space-time and satisfy $(\partial^k/\partial x^k) V(t,0)=0$ for every $t\in[0,t_{\rm max})$ and $k\geq 1$.

\begin{proof}
Suppose $V$ is smooth, with bounded derivatives that all equal zero at $x=0$. By Lemma~\ref{lem:condition_V_preserve}, the validity of condition~\eqref{eq:constraint} is equivalent to the statement that $V\psi\in D(H_0^{k-1})$ for every $\psi\in D(H_0^{k})$. With the given conditions on $V$, it is clear that  $V\psi$ has $2k$ square-integrable derivatives (on $(0,\infty)$ and $(-\infty, 0)$) if $\psi$ does.
Also, $(V\psi)^{(j)}(0)=V(0)\psi^{(j)}(0)$ for all $j<2k$, so $V\psi$ satisfies the same boundary conditions as $\psi$ at $x=0$. This shows that $V\psi\in D(H_0^{k})$ for every $\psi\in D(H_0^{k})$ and thus that~\eqref{eq:constraint} holds.

Now assume condition~\eqref{eq:constraint} holds.
Since functions in $D(H_0^k)$ are differentiable away from zero, one easily deduces that $V$ is too. 
Boundedness of these derivatives follows from the fact that $V$, viewed as an operator, is continuous (cf.~the proof of Lemma~\ref{lem:condition_V_preserve}).
From the boundary conditions it is clear that the derivatives of $\psi$ must have well-defined limits from the left and right at zero, and this again translates to $V$.
 
 A function $\psi\in \cI$, which is in $D(H_0^k)$ for every $k$, must fulfill infinitely many boundary conditions of the form
 \begin{equation}\label{eq:psi_bc}
 \left\lbrace\begin{aligned}
  \psi^{(2j)}(0^-)&= \psi^{(2j)}(0^+)\,,\\
  \psi^{(2j+1)}(0^+)-\psi^{(2j+1)}(0^-)&=\lambda \psi^{(2j)}(0)\,.
 \end{aligned}\right.
 \end{equation}
Then $V\psi$ must satisfy the same conditions, for every $\psi$ that satisfies~\eqref{eq:psi_bc}.
As we will now see, this implies that $V^{(k)}(0)=0$ for every $k\geq 1$. 
Consider the conditions~\eqref{eq:psi_bc} for $V\psi$ and $j=0$. 
Using continuity of $\psi$ at zero, the first condition yields continuity of $V$.
In view of the boundary condition for $\psi'$, the second condition implies
\begin{align*}
 0&=(V\psi)'(0^+)- (V\psi)'(0^-)-\lambda (V\psi)(0)\\
 &=\psi(0)\big(V'(0^+)- V'(0^-)\big)\,,
\end{align*}
so $V'$ is also continuous at zero. 
Now let $j\geq 1$ and assume we have continuity of $V^{(k)}$ at $x=0$ for $k\leq 2j-1$ and $V^{(k)}(0)=0$ for $0<k< 2j-1$.
Then 
\begin{multline*}
 (V\psi)^{(2j)}(0^+)=V(0)\psi^{(2j)}(0)+ V^{(2j)}(0^+)\psi(0) \\+ 2j V^{(2j-1)}(0)\psi'(0^+)\,.
\end{multline*}
Using the boundary condition for $\psi'$, the continuity of $(V\psi)^{(2j)}$, which is the first condition in~\eqref{eq:psi_bc}, gives
\begin{equation}\label{eq:jump_even}
 2j \lambda V^{(2j-1)}(0)= V^{(2j)}(0^-)-V^{(2j)}(0^+)\,.
\end{equation}
Now take $\psi$ satisfying~\eqref{eq:psi_bc} with $0=\psi(0)=\psi'(0)$ but $\psi''(0)\neq 0$. Then
\begin{multline*}
 (V\psi)^{(2j+2)}(0^+)= {2j+2 \choose 2}V^{(2j)}(0^+)\psi''(0)\\ + {2j+2\choose 3} V^{(2j-1)}(0)\psi^{(3)}(0^+) + V(0)\psi^{(2j+2)}(0)
\end{multline*}
and continuity of $(V\psi)^{(2j+2)}$ at $x=0$ implies
\begin{equation*}
 \frac{2j}{3} \lambda V^{(2j-1)}(0)= V^{(2j)}(0^-)-V^{(2j)}(0^+)\,.
\end{equation*}
Together with~\eqref{eq:jump_even} this shows that $V^{(2j-1)}(0)=0$, and that $V^{(2j)}$ is continuous at $x=0$. With this information we can use the second condition in~\eqref{eq:psi_bc}, with arbitrary $\psi\in \cI$, to obtain
\begin{equation*}
 2j \lambda V^{(2j)}(0)= V^{(2j+1)}(0^-)-V^{(2j+1)}(0^+)\,.
\end{equation*}
The same calculation for $(V\psi)^{(2j+3)}(0^+)-(V\psi)^{(2j+3)}(0^-)$, with $0=\psi(0)=\psi'(0)$, gives
\begin{equation*}
  \lambda\frac{2j (2j+1)}{3(2j+3)} V^{(2j)}(0)= V^{(2j+1)}(0^-)-V^{(2j+1)}(0^+)\,,
\end{equation*}
and this implies that $V^{(2j)}(0)=0$ and $V^{(2j+1)}$ is continuous at $x=0$. 
We can thus conclude that $V^{(k)}(0)=0$, for all $k\geq 1$, by induction.
\end{proof}

\subsection*{The hydrogen atom}

In this section we extend the previous 1D considerations to the case of the 3D hydrogen atom in radial external potentials. 
The corresponding operator is
$$H_0=-\Delta -\frac{1}{|\bx|}.$$
The equivalent of Theorem~\ref{thm:delta_1D} is the following.
\begin{theorem}[Radial potentials for the hydrogen atom]\label{thm:hydrogen}
Let $V=V(r)$ be a smooth radial potential which satisfies the condition~\eqref{eq:constraint} for the hydrogen atom.
Then at the origin $V$ satisfies 
$$ \frac{{\rm d}^k}{{\rm d}r^k}V(0)=0,\quad\text{for all } k\geq1.$$
\end{theorem}

The proof goes along the same lines as that of the corresponding implication in Theorem~\ref{thm:delta_1D}, but the calculations are more tedious and thus given in Appendix~\ref{app:hydrogen}. The converse implication as in Theorem~\ref{thm:delta_1D} also holds in this case, and a similar result is probably true for a non-radial potential, but we have not pursued in this direction. 

One could think that $\cV$ is so small because we chose $\cI$ too large. However, the proofs of Theorems~\ref{thm:delta_1D} and~\ref{thm:hydrogen} teach us that this effect occurs whenever $\cI$ contains functions with some non-vanishing derivatives at the singularity. This property holds for eigenfunctions of $H_0$, and the proof in Appendix~\ref{app:hydrogen} uses this for the first two eigenfunctions only. 

The set $\cV$ characterized by Theorems~\ref{thm:delta_1D} and~\ref{thm:hydrogen} is a subset of the one we considered in Section~\ref{sec:RG-smooth}, because it contains the additional hypothesis that derivatives must vanish at $x=0$ (respectively $r=0$ for hydrogen). For this set one can prove a Runge-Gross uniqueness theorem for $N=1$, or $N\geq 1$ with $w=0$, along the same lines as Theorem~\ref{thm:RG_smooth}. We omit the details, as this theorem for non-interacting systems and a small set $\cV$ seems to be of limited interest.

\subsection{The two-body Coulomb interaction}\label{sec:Coulomb}

In the previous section we have considered singular one-body operators and we have discovered that the set $\cV$ only contains functions that are very flat at the singularity.
In the two-body case, the situation is even worse. There is essentially no way for an external potential to avoid the singularity of the two-body electronic repulsion. For the singlet state Helium atom, we are actually able to prove that only constants remain in $\cV$, even if we only assume that $D(H_V^k)=D(H_0^k)$ for $k\leq 4$. 

\begin{theorem}[Potentials for two electrons]\label{thm:Helium}
Let
$$H_0=-\Delta_{\bx_1}-\Delta_{\bx_2}+\frac{1}{|\bx_1-\bx_2|}$$
be the Hamiltonian for two particles, restricted to permutation-symmetric, square-integrable functions.
Let $V(\bx)$ be an external potential that has 6 continuous bounded derivatives. If 
$D(H_V^4)=D(H_0^4)$, then $V$ is constant. 
\end{theorem}

The idea is to work in relative and center of mass coordinates and to use arguments as in the one-body case, in the direction $\bv=\bx_1-\bx_2$. We can then prove that 
$$\Delta_\bv (V(\bu+\bv/2) +V(\bu-\bv/2))\vert_{\bv=0}=\frac12\Delta V(\bu)=0$$
for all $\bu\in\R^3$, and this implies that $V$ is constant. The details of the proof are given in Appendix~\ref{sec:proof_Helium} below.
A similar result holds for triplet states but we will not treat this in detail. Also, we believe that the same result holds for $N$ electrons in appropriate symmetry classes. 

This shows that high-order Taylor expansions cannot be employed for Coulomb systems, and a different route has to be found. By a closer investigation of the time derivatives, we can get a weaker Runge-Gross Theorem that is still somewhat reasonable for practical purposes. More precisely, for smooth external potentials we can show that the first 4 time derivatives of $V_1-V_2$ are constant if the densities match.

Theorem~\ref{thm:Helium} suggests that it is not possible to differentiate $\rho(t,\bx)$ more than 3 times at $t=0$. We are going to use a trick that will allow us to differentiate it 5 times (but probably not more). The trick is to take a smooth test function $\phi$ and to differentiate
$$\int_{\R^3}\rho(t,\bx)\,\phi(\bx)\,\dx,$$
as was used in the proof of the Runge-Gross Theorem in Section~\ref{sec:RG-smooth}. This means that we are viewing the $t$-derivatives of $\rho(t,\bx)$ as distributions, or generalized functions. 

\begin{theorem}[Finite-order Runge-Gross for Coulomb systems]\label{thm:Coulomb}
Assume that $w(\bx-\by)=|\bx-\by|^{-1}$ is the Coulomb repulsion and that $V_0$ is a fixed external potential that has 6 bounded space derivatives. Let 
$$H_0=\sum_{j=1}^N-\Delta_{\bx_j}+V_0(\bx_j)+\sum_{1\leq j<k\leq N}\frac{1}{|\bx_j-\bx_k|}$$
be the corresponding Hamiltonian for $N$ particles, restricted to square-integrable functions that are symmetric (resp. anti-symmetric) under permutation of the particles. 
Let finally $\Psi_0\in D(H_0^4)$.

Assume that $V_1(t,\bx)$ and $V_2(t,\bx)$ are two potentials with 6 bounded space-time derivatives. Then the corresponding densities $\rho_1(t,\bx)$ and $\rho_2(t,\bx)$ have 5 (resp. 6 in the anti-symmetric case) time-derivatives in the sense that
$$f_m(t)=\int_{\R^3}\rho_m(t,\bx)\,\phi(\bx)\,\dx$$ 
is differentiable for every smooth $\phi(\bx)$.

If $\rho_1(t,\bx)=\rho_2(t,\bx)$ for all $t\in[0,t_{\rm max})$ and all $\bx\in\R^3$, then 
\begin{equation*}
\rho(0,\bx)\nabla\frac{\partial^\ell}{\partial t^\ell}(V_1-V_2)(0,\bx)=0
 \label{eq:identification_derivatives_Coulomb}
\end{equation*}
for all $\ell\leq 3$ (resp. $\ell\leq 4$ in the anti-symmetric case).
\end{theorem}

Due to its rather technical nature, we present the proof of this result in Appendix~\ref{sec:proof_Coulomb}.

\section{Discussion and outlook}

In this paper, we have analyzed in detail the method of Taylor expansions for the Runge-Gross Theorem. We have introduced an abstract setting, based on two sets: $\cI$ for the initial conditions $\Psi_0$'s and $\cV$ for the time-dependent external potentials $V(t,\bx)$. The choice of these sets guarantees that the density $\rho(t,\bx)$ is differentiable in time infinitely often and then the original Runge-Gross approach works well. The main question is to identify these two sets in practical situations.

Assuming that the potentials $V_0$ and $w$ are smooth, we found in Section~\ref{sec:RG-smooth} that the sets $\cI$ and $\cV$ also consist of smooth and bounded functions, without further restrictions. This is the correct mathematical setting for the original Runge-Gross Theorem. 

We then studied the case of singular potentials, as is relevant for physical applications. We found that singularities have very different consequences for $N=1$ and $N\geq 2$. In the one-particle case, the class of allowed potentials $V(t,\bx)$ is reduced to those that avoid the singularities, in the sense that all of their derivatives vanish there. This is a very small set, whose physical interest is debatable.

On the other hand, a singularity in the two-body potential cannot be avoided by an external potential. 
For the Coulomb interaction, the sole constraint on $\cI$ and $\cV$ that one can differentiate many times in $t$ already imposes that the external potential is constant, without knowing anything about the density. Therefore, high-order Taylor expansions are not the right tool to study atoms and molecules in TDDFT. As we have shown in Theorem~\ref{thm:Coulomb}, low-order expansions in $t$ can be used, but of course they yield only limited information. 

A natural strategy to avoid Taylor-expansions is to use the density $\rho(t,\bx)$ for all times and not just at $t=0$. One way to make $V(t,\bx)$ appear in an equation is to differentiate $\rho(t,\bx)$ only twice. This is possible for Coulomb interactions as we have shown in Theorem~\ref{thm:Coulomb}. This gives an implicit equation for $V(t,\bx)$~\cite{RugPenBau-09,RugLee-11,PenRug-11,RugGiePenLee-12,RugPenLee-15}. Unfortunately, this equation involves space derivatives of $\Psi(t,\bX)$ of higher order which are difficult to control. Hence standard techniques of functional analysis cannot be used in this context and new ideas are needed. 

\bigskip
\noindent\textbf{Acknowledgment.}
The authors acknowledge finacial support from the Danish Council for Independent Research (S.F., Sapere Aude Grant number DFF--4181-00221) and the European Research Council (J.L. and M.L.) under the European Community's Seventh Framework Programme (FP7/2007-2013 Grant Agreement MNIQS 258023). 
M.L. would like to thank Eric Cancès for stimulating discussions.
Part of this work was carried out while S.F. was invited professor at Universit\'e de Paris-Dauphine.
\appendix
\section{Proof of Theorem~\ref{thm:hydrogen}}\label{app:hydrogen}
%
To prove this theorem it is sufficient to consider radial functions $\psi$ that are elements of $D(H_0^k)$ for every $k$. As is well known, multiplication by $r$ is a unitary map from radial square-integrable functions on $\R^3$ to square-integrable functions on $[0,\infty)$.
Under this transformation, the operator $H_0$ becomes
\begin{equation*}
 h_0=-\frac{{\rm d}^2}{{\rm d}r^2}-\frac1r\,,
\end{equation*}
with domain $D(h_0)$ given by square-integrable functions $\psi(r)$ for which
\begin{equation*}
 \left\lbrace 
 \begin{aligned}
&\int_0^\infty \vert \psi''(r)\vert^2 {\rm d}r \text{ is finite~\footnotemark[1]}\\   
&\psi(0)=0\,.
\end{aligned}\right.
\end{equation*}
Thus, if $\psi$ is in $D(h_0^k)$ it must satisfy $h_0^{k-1}\psi(0)=0$. Now assume that $V(r)$ is a potential that maps $D(h_0^{k+1})$ to $D(h_0^{k})$ for every $k$. For $k=2$ this implies, using that $\psi(0)=0$ and $h_0\psi(0)=0$,
\begin{equation*}
 0=h_0V\psi(0)=\Big([h_0, V]\psi\Big)(0)=-2 V'(0)\psi'(0)\,,
\end{equation*}
and thus $V'(0)=0$. Now let
\begin{equation*}
 \psi(r)=\tfrac{16}{3} r\left(e^{-r/2}-(1-r/4)e^{-r/4}\right)\,.
\end{equation*}
This is just a multiple of the difference of the ground state of $h_0$ and the first excited state, so it is certainly an analytic vector of $h_0$. Observe also that $\psi'(0)=\psi''(0)=0$, but $\psi^{(3)}(0)=1$ and $(h_0\psi)'(0)=-1$ are different from zero. Now since $V\psi(r)$ and $Vh_0\psi(r)$ are elements of $D(h_0^k)$, $V(r)$ must satisfy the following equations, for every $k\geq 0$,
\begin{equation}\label{eq:system V(0)}
 \Big(h_0^{k+1}V\psi\Big)(0)=0\,, \qquad \Big(h_0^{k}Vh_0\psi\Big)(0)=0.
\end{equation}
This is an infinite system of linear equations for the derivatives of $V$ at $r=0$. We will exploit that this system is triangular in the pairs~\eqref{eq:system V(0)}. That is, we prove that $V^{(j)}(0)=0$ by induction. Assume that we already know that $V^{(j)}(0)=0$ for $1\leq j\leq 2k-3$, as we do for $k=2$. Then the equations~\eqref{eq:system V(0)} depend only on the values of $V^{(2k-2)}(0)$ and $V^{(2k-1)}(0)$, and~\eqref{eq:system V(0)} can be written as $A(V^{(2k-1)}(0), V^{(2k-2)}(0))=0$, for some $2\times 2$ matrix $A$, that we will now determine. We will show that $\det(A)\neq 0$, and so the derivatives of $V$ need to be zero.

To obtain $A$ from the equations~\eqref{eq:system V(0)} first note that if $f(r)$ is a smooth function with Taylor expansion at zero given by 
\begin{align}\label{eq:Taylor1}
 f(r) = a r^{2n-1} + b r^{2n} + O(r^{2n+1}),
\end{align}
then, assuming $n\geq 2$, $-h_0 f$ has Taylor expansion,
$$
-h_0 f(r) = \tilde{a} r^{2n-3} + \tilde{b} r^{2n-2} + O(r^{2n-1})\,,
$$
with $(\tilde{a}, \tilde{b}) = T_n (a,b)$, using the $2\times 2$ matrix
$$
T_n = \left( \begin{matrix} (2n-1)(2n-2) & 0\\ 1 & 2n (2n-1)
\end{matrix}\right)\,.
$$
Also, in the case $n=1$, we evaluate
$$
-h_0 f(0) = (1,2)\cdot(a,b)\,.
$$
Thus, if $f$ is given by \eqref{eq:Taylor1} for some $n$, then
\begin{align}
h_0^n f(0) = (-1)^n \left(T_n^*\cdots T_2^* (1,2) \right) \cdot (a,b)\,.
\end{align}
One easily checks that
\begin{align}\label{eq:matrixprod}
T_n^*\cdots T_2^* (1,2) = \left((2n-1)! \sum_{j=1}^n\frac{1}{2j-1}, (2n)! \right)\,,
\end{align}
by recursion.

Now, by the induction hypothesis, $V\psi(r)$ has a Taylor expansion of the form
\begin{multline*}
V \psi(r) = \frac{v}{3! (2k-2)!} r^{2k+1}\\ + 
\left( \frac{u}{3! (2k-1)!} - \frac{4}{3} \frac{v}{4! (2k-2)!}
\right)r^{2k+2} + O(r^{2k+3}),
\end{multline*}
where $(v,u)=(V^{(2k-2)}(0), V^{(2k-1)}(0))$.
Using this, together with formula~\eqref{eq:matrixprod} for $n=k+1$, the left hand side of
the first equation in (A1) becomes (after multiplication by $(-1)^{k+1}$)
\begin{align*}
v \left\{\binom{2k+1}{3}\sum_{j=1}^{k+1}\frac{1}{2j-1} - \frac{4}{3}\binom{2k+2}{4}\right\}
+ u \binom{2k+2}{3}\,. 
\end{align*}
Similarly, the second equation in~\eqref{eq:system V(0)} yields
\begin{align*}
-v(2k-1)\left\{ k - \sum_{j=1}^k \frac{1}{2j-1}\right\} + u 2k = 0.
\end{align*}
The determinant of this system simplifies to
\begin{multline*}
 {2k+1 \choose 3} \left\{ \frac23 k(k+1) +\frac{2k}{2k+1} - 2\sum_{j=1}^{k}\frac{1}{2j-1} \right\}\\
 \geq  {2k+1 \choose 3}\left\{\frac23 k(k+1)+\frac{2k}{2k+1} -2k\right\}\,.
\end{multline*}
This is strictly positive for $k\geq2$, so we find that $V^{(2k-2)}(0)=V^{(2k-1)}(0)=0$ is the only solution to~\eqref{eq:system V(0)}. This completes the proof.\qed
%
%
%
\section{Proof of Theorem~\ref{thm:Helium}}\label{sec:proof_Helium}
We express the operator $H_V$ with respect to the relative coordinate $\bv=\bx_1-\bx_2$ and center of mass coordinate $\bu=\tfrac12(\bx_1+\bx_2)$ :
\begin{equation*}
 H_V=-\tfrac12 \Delta_\bu - 2\Delta_\bv + \frac{1}{|\bv|} + V\left(\bu+\tfrac12\bv\right)+V\left(\bu-\tfrac12\bv\right)\,.
\end{equation*}
For simplicity, we will denote 
\begin{equation*}
 W(\bu, \bv)=V\left(\bu+\tfrac12\bv\right)+V\left(\bu-\tfrac12\bv\right)\,.
\end{equation*}
The domain of $H_V$ equals that of the pure kinetic energy operator $D(H_V)=D(T)=D(H_0)$, as follows from the criterion discussed in Lemma~\ref{lem:condition_V_preserve}. 
Although an element $\psi(\bu,\bv)$ of this space need not be continuous, it can be restricted to the hyperplane $\bv=0$ using the theory of Sobolev, yielding a square integrable function of $\bu$ (see~\cite[Theorem~IX.38]{ReeSim2}). As a consequence, $\psi(\bu,\bv)$ satisfies
\begin{equation}\label{eq:trace_cont}
 \lim_{\bv \to 0} |\bv|\, \psi(\bu,\bv)=0.
\end{equation}
If $V$ is not constant, this property will lead to a contradiction to $D(H_0^4)=D(H_V^4)$, because the latter implies that $H_V^3 W \psi \in D(T)$, but this diverges at $\bv=0$ leading to a non-zero limit in~\eqref{eq:trace_cont}.

We now give the details of this argument. First, we will see that with the given conditions on $V$, $D(H^k)=D(H_0^k)$ holds for $k=2,3$.
For $k=2$ we need to show that $\psi \in D(T)$ satisfies $H_0\psi\in D(H_0)=D(T)$ if and only if $H_V\psi\in D(H_V)=D(T)$. This follows from the fact that $W$ has two bounded derivatives and thus maps $D(T)=D(H_0)$ to itself.

The domain of the third power is given by those $\psi$ for which $H_V\psi\in D(H_V^2)=D(H_0^2)$. An element $\psi$ of $D(H_0^3)$ is thus in $D(H_V^3)$ if $H_0^2 W\psi(\bu,\bv)$ is square-integrable. 
Because $D(H_0^2)=D(H_V^2)$ and $W$ maps $D(H_0)$ to itself, we know that $H_0 W H_0 \psi(\bu,\bv)$ and $W H_0^2 \psi(\bu,\bv)$ are square-integrable. Thus, we need to show that
\begin{align*}
 &\big[H_0, [H_0, W]\big]\psi
 =\big[H_0, [-2\Delta_\bv-\tfrac12\Delta_\bu, W]\big]\psi\\
 &=\big[-2\Delta_\bv-\tfrac12\Delta_\bu, [-2\Delta_\bv-\tfrac12\Delta_\bu, W]\big] \psi\\
 &\qquad-4 (\nabla_{\bv} W)\cdot \frac{\bv}{|\bv|^3}  \psi
\end{align*}
is square integrable, for $\psi\in D(H_0^3)$.
For the first term this follows from the differentiability of $W$ and the fact that $\psi\in D(T)$ has two square-integrable derivatives. For the second term, note that $(\nabla_{\bv} W)(\bu,0)=0$, so $(\nabla_{\bv} W)\cdot \bv/|\bv|^3$ diverges like $1/|\bv|$, i.e. like the Coulomb-potential, which is well-defined on $D(T)$.
This shows that $D(H_0^3)\subset D(H_V^3)$ (the converse inclusion is shown in the same way and plays no role in our argument).

Now let $\psi(\bu,\bv)\in D(H_0^4)$ be a function of the form $f(\bu)g(\bv)$ with smooth $f(\bu)$ and $g(\bv)=e^{|\bv|/4}$ for $|\bv|<1$. This is clearly a possible choice, as $g(\bv)$ is an eigenfunction of $-2\Delta_\bv+1/|\bv|$ for $|\bv|<1$, where the singularity lies. We will see that all the functions of this type are in $D(H_V^4)$ only if $V$ is constant.
Assume that this function $\psi(\bu,\bv)$ is also an element of $D(H_V^4)$. We then have
$H_V\psi\in D(H_V^3)=D(H_0^3)$ and can argue as in the previous step that
\begin{equation*}
 \big[H_0, [H_0, W]\big]\psi
 =\big[H_0, [-2\Delta_\bv-\tfrac12\Delta_\bu, W]\big]\psi
 \in D(T)\,.
\end{equation*}
Since the chosen $\psi(\bu,\bv)$ is smooth in the variable $\bu$, it is easy to see that terms in this commutator that involve derivatives in this variable satisfy~\eqref{eq:trace_cont}. We deduce that
\begin{equation*}
 \lim_{\bv\to 0} |\bv|\,\big[- 2\Delta_\bv + \frac{1}{|\bv|}, [-2\Delta_\bv, W(\bu, \bv)]\big]\psi(\bu,\bv)=0
\end{equation*}
must hold. This commutator evaluates to
\begin{multline*}
 4(\Delta\Delta W)\psi+16 (\nabla \Delta W) \cdot\nabla \psi\\
 +16\tr(\mathrm{Hess}(W)\mathrm{Hess}(\psi))- 4(\nabla W) \cdot \frac{\bv}{|\bv|^3}  \psi\,,
\end{multline*}
where all the derivatives are taken only in the variable $\bv$. If we choose a unit vector $\omega$ such that $\bv=|\bv| \omega$, the limit $\bv\to 0$ certainly remains unchanged if we average over this variable.
The first two terms disappear in the limit because $\psi$, $\nabla \psi$, and the derivatives of $W$ are bounded. We conclude that
\begin{multline}
 \lim_{|\bv|\to 0} \frac{|\bv|}{4\pi} \int\Big(
 4\tr(\mathrm{Hess}(W)\mathrm{Hess}(\psi))\\- (\nabla W) \cdot \frac{\bv}{|\bv|^3}  \psi\Big) \mathrm{d}\omega=0\,.
 \label{eq:truc}
\end{multline}
To calculate the integral of the second term, we perform a Taylor expansion of $\nabla W(\bu,\bv)$ at $\bv=0$ (where it vanishes), and find
\begin{align*}
 &\lim_{|\bv|\to 0} \frac{|\bv|}{4\pi}\int(\nabla W) \cdot \frac{\bv}{|\bv|^3}  \psi(\bu, \bv) \mathrm{d}\omega\\
 &\qquad=\lim_{|\bv|\to 0} \frac{1}{4\pi}\int\psi(\bu, \bv)\, \frac{\nabla W(\bu,\bv)}{|\bv|}\cdot \omega\,  \mathrm{d}\omega\\
 &\qquad=\frac{\psi(\bu, 0)}{4\pi}\int \omega\cdot\mathrm{Hess}(W)(\bu, 0)\omega \,\mathrm{d}\omega\\
 &\qquad=\tfrac13 f(\bu) \Delta_{\bv}W(\bu, 0)\,,
\end{align*}
where we have used that
\begin{equation*}
 \frac{1}{4\pi}\int \omega_i\omega_j \,\mathrm{d}\omega = \tfrac13 \delta_{ij}\,.
\end{equation*}
For the first term in~\eqref{eq:truc} we use the explicit form
\begin{equation*}
 \mathrm{Hess}(\psi)_{ij}=f(\bu)e^{|\bv|/4}\left(\frac{v_iv_j}{16|\bv|^2} - \frac{v_i v_j}{4|\bv|^3}+\frac{\delta_{ij}}{4|\bv|} \right)\,,
\end{equation*}
to obtain
\begin{align*}
 &\lim_{|\bv|\to 0} \frac{|\bv|}{4\pi} \int 4\tr(\mathrm{Hess}(W)\mathrm{Hess}(\psi)) \mathrm{d}\omega\\
 &= f(\bu)\Delta_{\bv}W(\bu, 0)\\
 &\qquad- \lim_{|\bv|\to 0} \frac{f(\bu)}{4\pi}\int \omega\cdot\mathrm{Hess}(W)(\bu, \bv) \omega\, \mathrm{d}\omega\\
 &=\tfrac23 f(\bu) \Delta_{\bv}W(\bu, 0)\,.
\end{align*}

Since $\Delta_{\bv}W(\bu, 0)=\tfrac12 \Delta V(\bu)$ this adds up to the conclusion that
\begin{equation*}
 \lim_{|\bv|\to 0} \frac{|\bv|}{4\pi}\int (H_0^2W\psi)(\bu, \bv)\,\mathrm{d}\omega=\tfrac23 f(\bu) \Delta V(\bu)=0.
\end{equation*}
Since the function $f(\bu)$ was arbitrary, this means that $V$ must be harmonic. Since it was also assumed to be bounded, $V$ is constant.\qed

\section{Proof of Theorem~\ref{thm:Coulomb}}\label{sec:proof_Coulomb}
In this section we prove Theorem~\ref{thm:Coulomb} by studying precisely for which $k$ we have $D(H_0^k)=D(H_V^k)$. From Theorem~\ref{thm:Helium} we know that this does not hold for $k\geq 4$. 
To obtain sharper results, we will use the concept of the domain of a half-integer power of an operator. For the case we treat here, the following may serve as a definition. The domain $D(T^{1/2})$ consists of those square-integrable $\Psi(\bX)$ for which
\begin{equation*}
 \int_{\R^{3N}} |\nabla_\bX \Psi(\bX)|^2\mathrm{d}\bX = \int_{\R^{3N}} |\mathbf{K}|^2 |\widehat{\Psi}(\mathbf{K})|^2 \mathrm{d}\mathbf{K}
 \end{equation*}
is finite. 
We have $D(T^{1/2})=D(H_0^{1/2})=D(H_V^{1/2})$ and define recursively $D(H_0^{k+1/2})$ to be those $\Psi\in D(H_0^{1/2})$ such that $H_0 \Psi \in D(H_0^{j+1/2})$ for every $j<k$, and analogously for $H_V$.

\begin{lemma}\label{lem:DH^3}
Let $H_0$ be the Hamiltonian with Coulomb interaction $w(\bx-\by)=1/|\bx-\by|$ of Theorem~\ref{thm:Coulomb} and let $V(\bx)$ be a function with four bounded derivatives on $\R^3$. If $\Psi$ is symmetric or anti-symmetric, then
\begin{equation}\label{eq:ad^2}
 \bigg(\sum_{j=1}^N \big[H_0, [H_0, V(\bx_j)]\big]\Psi\bigg)(\bx_1, \dots, \bx_N)
\end{equation}
is square-integrable for every $\Psi\in D(H_0)$ and we have $D(H_0^{3})=D(H_V^{3})$.

If $\Psi$ is anti-symmetric, we additionally have that $|\bx_i-\bx_j|^{-2}\Psi(\bx_1, \dots, \bx_N)$ is square-integrable and $\Psi \in D(T^{3/2})$ for every $\Psi\in D(H_0^{3/2})$, as well as $D(H_0^{7/2})=D(H_V^{7/2})$.
\end{lemma}

This lemma tells us that for a sufficiently smooth potential $V(\bx)$, we have in the symmetric case $D(H_0^{3})=D(H_V^{3})$. In Theorem~\ref{thm:Helium} we have shown that $D(H_0^{4})\neq D(H_V^{4})$, except if $V$ is constant. However, the same proof also shows that $D(H_0^{7/2})\neq D(H_V^{7/2})$. In the anti-symmetric case, we have $D(H_0^{7/2})= D(H_V^{7/2})$ and a reasoning similar to that of Theorem~\ref{thm:Helium} shows that $D(H_0^{4})\neq D(H_V^{4})$ as well.

\begin{proof}
 First, calculate the commutator
 \begin{align}
  &\sum_{j=1}^N \big[H_0, [H_0, V(\bx_j)]\big]\label{eq:double_comm}\\
  &\quad=\sum_{j=1}^N \Big(\sum_{k=1}^N \big[-\Delta_{\bx_k} + V_0(\bx_k), [-\Delta_{\bx_j}, V(\bx_j)]\big]\notag\\
  &\quad\quad   +\sum_{1\leq \ell<k\leq N} \big[|\bx_\ell-\bx_k|^{-1}, [-\Delta_{\bx_j}, V(\bx_j)]\big]\Big)\notag
  \,.
 \end{align}
When the first term acts on $\Psi(\bX)$ the result is square-integrable, because $\Psi\in D(T)$ has two square-integrable derivatives. The second commutator equals
\begin{align*}
 &\big[|\bx_\ell-\bx_k|^{-1}, [-\Delta_{\bx_j}, V(\bx_j)]\big]\\
 &\qquad=2\,\nabla_{\bx_j}|\bx_\ell-\bx_k|^{-1}\cdot\nabla_{\bx_j}V(\bx_j)\\
 &\qquad=\nabla_{\bx_j}V(\bx_j)\cdot\Big(\delta_{j\ell} \frac{\bx_\ell-\bx_k}{|\bx_\ell-\bx_k|^3}
 -\delta_{jk} \frac{\bx_\ell-\bx_k}{|\bx_\ell-\bx_k|^3}\Big)\,.
\end{align*}
Taking the sum over $j$ then gives
\begin{equation*}
 \frac{\bx_\ell-\bx_k}{|\bx_\ell-\bx_k|^3}\cdot\Big(\nabla V(\bx_\ell)-\nabla V(\bx_k)\Big)\,.
\end{equation*}
Since $V(\bx)$ has bounded second derivatives, the function above is smaller in modulus than $|\bx_\ell-\bx_k|^{-1}$ times a bounded function. This implies that $\Psi\in D(T)$ multiplied by this function is square-integrable.
Thus the action of~\eqref{eq:double_comm} on $\Psi\in D(H_0^3)$ yields a square-integrable function, which implies $\Psi\in D(H_V^3)$ as argued in the proof of Theorem~\ref{thm:Helium}.
It is obvious from these calculations that the same holds for the double commutator of $\sum_{j=1}^N V(\bx_j)$ with $H_V$ and we conclude that $D(H_0^{3})=D(H_V^{3})$.

Now let $\Psi(\bX)$ be anti-symmetric, so that it vanishes if $\bx_i=\bx_j$. 
We will show that then $|\bx_i-\bx_j|^{-2}\Psi(\bX)$ is square-integrable for $\Psi\in D(T)$.
To see why this implies that any $\Psi\in D(H_0^{3/2})$ is an element of $D(T^{3/2})$, first write $T=H_0 - \sum_{j=1}^N V_0(\bx_j) - \sum_{1\leq j<k\leq N} |\bx_j-\bx_k|^{-1}$.
Then note that $H_0\Psi\in D(T^{1/2})$, by definition, and $\sum_{j=1}^N V_0(\bx_j)\Psi\in D(T^{1/2})$, by regularity of $V_0(\bx)$.
We are left with the interaction term, and it remains to check that the partial derivatives of $\sum_{1\leq j<k\leq N} |\bx_j-\bx_k|^{-1}\Psi(\bX)$ are square-integrable. The derivatives of $\Psi(\bX)$ define elements of $D(T^{1/2})$, and multiplication by $|\bx_j-\bx_k|^{-1}$ yields a square-integrable function (see~\cite[p.~169]{ReeSim2}). Finally, if a derivative acts on $|\bx_j-\bx_k|^{-1}$, the resulting term is bounded in modulus by a constant times $|\bx_j-\bx_k|^{-2}\Psi(\bX)$.

Now, back to the square-integrability of $|\bx_i-\bx_j|^{-2}\Psi(\bX)$. 
%
%
%
%
%
We will prove this by showing an inequality for a function $f(\by)$ of the relative coordinate $\by=\bx_i-\bx_j$, which can then be integrated over the remaining coordinates to obtain square-integrability of $|\bx_i-\bx_j|^{-2}\Psi(\bX)$.
Using anti-symmetry, we may write $f(\by)=\frac{\by}{2}\cdot \int_{-1}^1 (\nabla f)(t\by) {\rm d}t$, and the Minkowski inequality~\cite[Theorem~2.4]{LieLos-01} gives
\begin{align*}
\int_{\R^3} \frac{|f(\by)|^2}{|\by|^4} \dy
&\leq \int_{\R^3} \frac{1}{4|\by|^2} \left|\int_{-1}^1 (\nabla f)(t\by) {\rm d}t \right|^2 \dy\\
&\leq \left(\int_{-1}^1 \left(\int_{\R^3} \frac{|(\nabla f)(t\by)|^2}{4|\by|^2}\dy \right)^{1/2}\dt\right)^{2}\,.
\end{align*}
Hardy's inequality~\cite[p.~169]{ReeSim2} implies
\begin{equation*}
 \int_{\R^3} \frac{|(\nabla f)(t\by)|^2}{4|\by|^2}\dy\leq \sum_{k,\ell=1}^3 \int_{\R^3} t^2 |(\partial_{y_k}\partial_{y_\ell}f)(t\by)|^2\dy\,.
\end{equation*}
Changing variables $\bz=t\by$ and performing the $t$-integral then gives
\begin{align*}
 \int_{\R^3} \frac{|f(\by)|^2}{|\by|^4} \dy&\leq 4
 \int_{\R^3}\sum_{k,\ell=1}^3 |(\partial_{y_k}\partial_{y_\ell}f)(\bz)|^2 \dz \,.
\end{align*}
This shows that $|\by|^{-2}f(\by)$ is square-integrable if $f(\by)$ has two square-integrable derivatives, and thus $|\bx_i-\bx_j|^{-2}\Psi(\bX)$ is square-integrable for $\Psi \in D(T)$.

To prove the inclusion of $D(H_0^{7/2})$ in $D(H_V^{7/2})$, it remains to show that $H_0^2 H_V \Psi\in D(T^{1/2})$ for all $\Psi\in D(H_0^{7/2})$. By the arguments in the proof of Theorem~\ref{thm:Helium} this is equivalent to showing that
\begin{equation}\label{eq:ad^2_fermion}
 \nabla_{\bx_i} \sum_{j=1}^N \big[H_0, [H_0, V(\bx_j)]\big]\Psi(\bX)
\end{equation}
is square-integrable for any $i\leq N$. From the formulas for the commutator we see that the only non-trivial terms are
\begin{equation*}
 \nabla_{\bx_\ell}\frac{\bx_\ell-\bx_k}{|\bx_\ell-\bx_k|^3}\cdot \Big(\nabla V(\bx_\ell)-\nabla V(\bx_k)\Big)\Psi(\bX)\,,
\end{equation*}
but these can be controlled by combining all the bounds we have just discussed.
Again, the argument for the converse inclusion is the same.
\end{proof}
Coming back to the proof of Theorem~\ref{thm:Coulomb}, we first need to show that
 \begin{align*}
  f(t)&=\int_{\R^3}\rho(t,\bx)\,\phi(\bx)\,\dx\\
  &=N\langle \Psi(t)\vert \varphi(\bx_1) \Psi(t)\rangle
 \end{align*}
 has five (resp. six if $\Psi_0$ is anti-symmetric) derivatives for $V(t,\bx)$ satisfying the conditions on $V_m$, $m=1,2$ of Theorem~\ref{thm:Coulomb}. Since $\Psi_0\in D(H_0^4)$, Lemma~\ref{lem:DH^3} implies that $\Psi_0\in D(H_V^3)$, so we immediately have existence of three derivatives of $f(t)$. Using the (anti-)symmetry of $\Psi_0$, the third derivative equals
\begin{align*}
  &\frac{{\rm d}^3}{{\rm d} t^3} f(t)=N\Big\langle \Psi(t)\Big\vert -\big[\dot{V}(t,\bx_1),[H_0,\varphi(\bx_1) ]\big]\Psi(t)\Big\rangle\\
  & +i \sum_{j=1}^N \bigg( \Big\langle H_V \Psi(t)\Big\vert -\big[H_V,[H_V, \phi(\bx_j)]\big] \Psi(t)\Big\rangle\\
  &\qquad - \Big\langle \Psi(t)\Big\vert -\big[H_V,[H_V, \phi(\bx_j)]\big]H_V \Psi(t)\Big\rangle\bigg)\,.
\end{align*}
Let $A$ denote the operator 
\begin{equation*}
 A=-\sum_{j=1}^N \big[H_V,[H_V, \varphi(\bx_j)]\big]=-\sum_{j=1}^N \big[H_V,[H_0, \phi(\bx_j)]\big].
\end{equation*}
$A$ is a symmetric operator, so we can  write
\begin{align}\label{eq:third_deriv_weak}
 \frac{{\rm d}^3}{{\rm d} t^3} f(t)
 =&-2 \Im\Big(  \langle H_V \Psi(t)\vert A \Psi(t)\rangle\Big)\\
 &-2N\Big\langle \Psi(t)\Big\vert \big(\nabla \dot{V}(\bx_1)\cdot \nabla \phi(\bx_1) \big) \Psi(t)\Big\rangle\notag
\end{align}
Equation~\eqref{eq:ad^2} (with $V=\varphi$), together with differentiability of $V$, implies that $A$ maps functions in $D(T)$ to square-integrable functions. The same obviously holds for the first four time-derivatives of $A$. This shows that
$A\Psi(t)$ has two time-derivatives which are square-integrable functions.
Consequently, the expression~\eqref{eq:third_deriv_weak} can be differentiated twice in time, and $f(t)$ five times. 

Those terms in the fifth derivative of $f(t)$ that involve a time-derivative of $V$ are once again differentiable, by the argument above. Using the symmetry of $A$, the remaining terms can be brought into the form
\begin{equation}\label{eq:sixth_deriv}
 2 \Im\left(\langle H_V^3 \Psi(t)\vert A \Psi(t)\rangle -3\langle H_V^2 \Psi(t)\vert A H_V \Psi(t)\rangle\right)\,.
\end{equation}
The second term is clearly differentiable because $\Psi(t)\in D(H_0^3)$.
To treat the first term of~\eqref{eq:sixth_deriv}, note that, if $\Psi_0$ is anti-symmetric, $A \Psi_0$ is an element of $D(T^{1/2})=D(H_V^{1/2})$ by~\eqref{eq:ad^2_fermion} and we can write
\begin{align*}
 &\langle H_V^3 \Psi(t)\vert A \Psi_0\rangle\\
 &=\langle  (T+1)^{-1/2}H_V^3\Psi(t)\vert(T+1)^{1/2} A \Psi_0\rangle\,.
\end{align*}
The function $(T+1)^{-1/2}H_V^3\Psi(t)$ is differentiable in $t$ with square-integrable derivative, as follows from the argument of Theorem~\ref{thm:regularity} applied to the domain $D(H_0^{7/2})$ of $H_{V(t)}$ in the Hilbert space $D(H_0^{5/2})$.
Using this, one easily sees that the difference quotient for~\eqref{eq:sixth_deriv} has a limit, in the same way one proves the product-rule for the derivative. This shows that~\eqref{eq:sixth_deriv} is differentiable, and hence the sixth derivative of $f(t)$ can be taken.

We have thus shown that we can take the first five (resp. six) derivatives that were used in the proof of Theorem~\ref{thm:RG_smooth} and we can follow that proof up to this order.
The argument goes through essentially unchanged, though there are some subtleties regarding the regularity of the functions involved that we comment on below.

First of all, it is important to remark that $\rho(0,\bx)$ is a continuous function, although $\Psi_0$ might not be, because 
 the evaluation of $\Psi_0(\bx_1, \dots, \bx_N)$ at $\bx_1=\bx$ is a square-integrable function of $\bx_2,\dots, \bx_N$ that depends continuously on $\bx$ (this follows from~\cite[Theorem~IX.38]{ReeSim2}).
 Hence, the connected components of the set $\{\rho(0,\bx)>0\}$ are open, and
\begin{equation*}
 \int \rho(0,\bx) |\nabla (V_1-V_2)|^2(\bx) {\rm d}\bx=0
\end{equation*}
implies that $V_1(\bx)-V_2(\bx)$ is constant on every component.

Finally we need to show that for $k\leq 3$ (respectively $k\leq 4$ in the fermionic case)
\begin{multline*}
 \frac{{\rm d}^{k+2}}{{\rm d} t^{k+2}} \left(f_1(t)-f_2(t)\right)\Big\vert_{t=0}\\
 =-2N\Big\langle \Psi_0\Big\vert \big(\nabla \partial_t^{k}(V_1-V_2)(\bx_1)\vert_{t=0}\cdot \nabla \phi(\bx_1) \big) \Psi_0\Big\rangle\,,
\end{multline*}
where $f_m(t)$ denotes the function $f(t)$ with $V=V_m$, using that for $l<k$
\begin{equation*}
 \int_{\R^3} |\nabla \partial_t^\ell(V_1-V_2)| \rho(0,\bx) \dx=0\,.
\end{equation*}
In the proof of Theorem~\ref{thm:RG_smooth} this was shown using smoothness of $\Psi_0$, which is no longer given.
We will now discuss how to adapt these arguments for the fermionic case and $k=4$, which is the most difficult case.

The term with no `explicit' time-derivatives in $f^{(6)}_m$, i.e. the term without time-derivatives on $V_m$, is (cf.~\eqref{eq:sixth_deriv})
\begin{multline*}
g_m(t):=
2 \Re \big\{\langle H^4_{V_m} \Psi | A_m \Psi \rangle -4 \langle H^3_{V_m} \Psi | A_m H_{V_m} \Psi \rangle \\ + 3  \langle H^2_{V_m} \Psi | A_m H^2_{V_m} \Psi \rangle
\big\}.
\end{multline*}
Here $A_m$ denotes the operator $A$ with $V=V_m$ and the term $\langle H^4_{V_m} \Psi | A_m \Psi \rangle$ has to be understood in the sense of the pairing between $A_m\Psi(t)\in D(T^{1/2})$ and the distribution $H^4_{V_m} \Psi(t)$, as was done in showing existence of the sixth derivative above.
The other pairings are standard scalar products.

We will only analyze $g_m$ in detail and show that $g_1(0)-g_2(0)=0$. The other terms in $f^{(6)}$ have explicit time derivatives and therefore fewer operators. Therefore all scalar products are immediately understandable as integrals and also the algebra is slightly simpler---due to the fewer factors. However, the basic idea of the calculation is the same.

Write $W = V_2 - V_1$ and observe that since $\int \rho(0,\mathbf{x})| \nabla W|(\mathbf{x}) d\mathbf{x} = 0 $, we have
$$
\sum_{j=1}^N |\nabla W |(x_j) \Psi_0 = 0.
$$
Actually, even more is true, namely
\begin{align}\label{eq:vanishing-3}
\sum_{j=1}^N \left( \Big|  |\nabla W |(x_j) H_{V_1}^s\Psi_0 \Big| + \Big|  |\nabla W |(x_j) \nabla H_{V_1}^s\Psi_0 \Big|\right)= 0.
\end{align}
with $s \in \{1,2,3\}$. Of course, a similar identity holds for $V_2$, and one can even have mixed products of the two operators. 
As in the the proof of Theorem~\ref{thm:RG_smooth}, these identities follow from the fact that $H_{V_1}$ is a {\it local} operator, so that $H_{V_1}^s \Psi_0 = 0$ on the open set where $\nabla W$ does not vanish. Here we used that $H_{V_1}^s \Psi_0$ (and $\nabla H_{V_1}^s \Psi_0$) defines a function (in contrast to a distribution).

Similarly to~\eqref{eq:vanishing-3} we have
\begin{align}\label{eq:SameA-3}
A_1 \Phi = A_2 \Phi,
\end{align}
and
\begin{align}\label{eq:SameA2-3}
[W,A_m] \Phi = 0,
\end{align}
for $\Phi = \Psi_0$ or $\prod_{s=1}^k H_{V_{j_s}} \Psi_0$, with $k\leq 3$ and $j_s \in \{ 1,2\}$.
The equality~\eqref{eq:SameA-3} allows us to replace $A_2$ by $A_1$---which we abbreviate by $A$---everywhere in $g_2(0)$.
Notice also that using \eqref{eq:vanishing-3} we get for $k\leq 4$,
\begin{align*}
H_{V_2}^k\Psi_0=(H_{V_1} +W)^k \Psi_0 = \sum_{j\leq k} \binom{k}{j} W^j H_{V_1}^{k-j} \Psi_0.
\end{align*}

Therefore,
\begin{align*}
&g_2(0)-g_1(0)\\
&=
\begin{aligned}[t]&2\Re\Big\{
4 \langle W H_{V_1}^3 \Psi_0 | A \Psi_0\rangle + 6 \langle W^2 H_{V_1}^2 \Psi_0 | A \Psi_0\rangle\\
&+ 4 \langle W^3 H_{V_1} \Psi_0 | A \Psi_0 \rangle + \langle W^4 \Psi_0 | A \Psi_0 \rangle\\
&- 4 \langle (H_{V_1}^3 + 3W H_{V_1}^2 + 3 W^2 H_{V_1} + W^3) \Psi_0 | A W \Psi_0 \rangle\\
&- 4 \langle (3W H_{V_1}^2 + 3 W^2 H_{V_1} + W^3) \Psi_0 | A H_{V_1} \Psi_0 \rangle \\
&+3 \langle (2 W H_{V_1} + W^2) \Psi_0 | A (H_{V_1}^2 + 2 W H_{V_1} + W^2) \Psi_0 \rangle\\
&+ 3\langle H_{V_1}^2 \Psi_0 | A( (2 W H_{V_1} + W^2) \Psi_0 \rangle  \Big\}.
\end{aligned}
\end{align*}
At this point notice that all involved vectors belong to the Hilbert space $L^2$, i.e. there are no distributions anymore.
A tedious calculation shows that this term vanishes. Let us only argue that the terms with four powers of $W$ cancel, the argument for the others being similar (Notice though, that for the other powers of $W$, we use the fact that we take the real part to reach the conclusion). I.e. we will prove that
\begin{multline*}
0=\Re \Big\{\langle W^4 \Psi_0 | A \Psi_0 \rangle - 4 \langle W^3 \Psi_0 | A W \Psi_0 \rangle\\
+ 3 \langle W^2 \Psi_0 | A W^2 \Psi_0 \rangle \Big\}.
\end{multline*}
But this is easy, since using~\eqref{eq:SameA2-3}, we get
\begin{align*}
&\langle W^4 \Psi_0 | A \Psi_0 \rangle - 4 \langle W^3 \Psi_0 | A W \Psi_0 \rangle + 3 \langle W^2 \Psi_0 | A W^2 \Psi_0 \rangle \nonumber \\
&=(1-4+3)\langle W^2 \Psi_0 | A W^2\Psi_0 \rangle=0.\nonumber
\end{align*}

This concludes the proof of Theorem~\ref{thm:Coulomb}.\qed

%

\end{document}